\documentclass[11pt,draftcls,onecolumn]{IEEEtran}

\usepackage{amssymb}
\usepackage{amsmath}
\usepackage{amsfonts}
\usepackage{bm}
\usepackage{latexsym}
\usepackage{mathrsfs}
\usepackage{calc}
\usepackage{cite}

\newtheorem{thm}{Theorem}
\newtheorem{prop}[thm]{Proposition}
\newtheorem{lem}[thm]{Lemma}
\newtheorem{defn}[thm]{Definition}
\newtheorem{cor}[thm]{Corollary}
\newtheorem{rem}[thm]{Remark}
\newtheorem{example}[thm]{Example}

\newcommand{\pd}{\partial}

\newcommand{\re}{\mbox{Re }}
\newcommand{\im}{\mbox{Im }}

\newcommand{\C}{\mathbb C}
\newcommand{\R}{\mathbb R}
\newcommand{\E}{\mathbb E}
\newcommand{\p}{\mathbb P}
\newcommand{\B}{\mathcal B}
\newcommand{\wt}{\widetilde}
\newcommand{\wh}{\widehat}
\newcommand{\sgn}{\operatorname{sgn}}
\newcommand{\ep}{\varepsilon}

\begin{document}

\title{On the instantaneous frequency of Gaussian stochastic processes}

\author{Patrik Wahlberg and Peter J. Schreier,~\IEEEmembership{Senior Member,~IEEE}.
\thanks{P. Wahlberg is with Dipartimento di Matematica, Universit\` a di Torino, Via Carlo Alberto 10, 10123 Torino (TO), Italy. Contact data: ph +39 011 6702944, fax +39 011 6702878, email patrik.wahlberg@unito.it. P. Schreier is with the School of Electrical Engr \& Computer Science,
The University of Newcastle, Callaghan, NSW 2308, Australia.  Contact data: ph +61 2
49215997, fax +61 2 49216993, email Peter.Schreier@newcastle.edu.au.
This work was supported by the Australian Research Council (ARC)
under the Discovery Project DP0664365.}}

\maketitle

\begin{abstract}
This paper concerns the instantaneous frequency (IF) of continuous-time, zero-mean, complex-valued, proper, mean-square differentiable nonstationary Gaussian stochastic processes. We compute the probability density function for the IF for fixed time, which extends a result known for wide-sense stationary processes to nonstationary processes. For a fixed time the IF has either zero or infinite variance. For harmonizable processes we obtain as a byproduct that the mean of the IF, for fixed time, is the normalized first order frequency moment of the Wigner spectrum.
\end{abstract}


\section{Introduction}

This paper treats the instantaneous frequency (IF) of mean square differentiable Gaussian zero-mean complex-valued proper nonstationary stochastic processes defined on $\R$. The IF of a stochastic process is the derivative of the phase function and it is a stochastic process (real-valued) itself. Our main result is an explicit formula for the probability density function (pdf) for the IF for an arbitrary fixed time point. The pdf is parameterized by the covariance function of the process and its derivatives evaluated at the same time point. The time axis may be divided into three subsets with completely different IF process variance behavior: In the first subset, the pdf has heavy tails and behaves like $x^{-3}$ for large $x$, which means that the IF variance is infinite. In the second subset, the IF has a degenerate pdf consisting of a Dirac measure at its mean (that is, the IF is deterministic). In the third subset the IF is $+\infty$ with probability one. (Alternatively one may say that the IF is not defined on the third subset.)
For wide-sense stationary (WSS) processes we show that the first subset is either $\R$ or empty, and in the latter case the covariance function has real part $\rho_x(t)=\alpha \cos(\beta t)$ for $\alpha,\beta \in \R$, $\alpha>0$. For harmonizable but not WSS processes we show through examples that the first set may be $\R$ or it may be empty,
but not ruling out other possibilities.
The formula we obtain for the pdf of the IF for fixed time is a generalization of results on the pdf of the IF for Gaussian \textit{WSS} processes derived by Miller \cite{Miller1} and Broman \cite{Broman1}.

Our result is formulated for Gaussian zero-mean complex-valued proper processes that are mean-square differentiable and have mean-square continuous derivative. When we specialize to certain harmonizable stochastic processes $z(t)$, the formula for the pdf of the IF for fixed time implies the identity
\begin{equation}\label{iffrekmoment37}
\E \frac{d}{dt} \arg z(t) = \frac{\int_{\R} \xi \wt W_z(t,\xi) d\xi }{ \int_{\R} \wt W_z(t,\xi) d\xi }
\end{equation}
where $\wt W_z$ denotes the Wigner spectrum. The Wigner spectrum is the expected value of the
Wigner distribution $W_f$, which is defined by
\begin{equation}\label{wvd1}
W_f (t,\xi) = \int_{\R} f (t+\tau/2) \overline{f(t-\tau/2)} e^{-i \tau \xi} d\tau, \quad t,\xi \in \R,
\end{equation}
for a function $f: \R \mapsto \C$.

Research about the IF has a long history in communications, signal processing and time-frequency analysis \cite{Cohen1,Cramer1,Papoulis1}. In analog frequency modulation the IF (minus a constant carrier frequency) represents the information in a modulated signal \cite{Papoulis1}. In time-frequency analysis \cite{Cohen1} the IF of a deterministic signal $f$ has been related to the Wigner distribution $W_f$ by the formula
\begin{equation}\label{momentformel0}
\frac{d}{dt} \arg f(t) = \frac{\int_{\R} \xi W_f(t,\xi) d\xi }{ \int_{\R} W_f(t,\xi) d\xi }.
\end{equation}
If a deterministic function is interpreted as a degenerate stochastic process, then this is a special case of \eqref{iffrekmoment37}.

The Wigner distribution gives a time-frequency description of $f$. It was introduced in Quantum Mechanics as a candidate for a pdf of a particle in phase space. It is well known that $W_f$ is only rarely a nonnegative function \cite{Grochenig1} so $W_f$ may not be interpreted as an energy distribution or as a pdf in general. This is consistent with the Uncertainty Principle \cite{Grochenig1} which gives upper bounds on the resolution of phase space localization of particles, or, in the signal analysis interpretation, of the time-frequency resolution of signals. Nevertheless, if the Wigner distribution $W_f$ is convolved by a sufficiently wide Gaussian function it becomes nonnegative, so domains of sufficiently large area in the time-frequency (phase) plane admits localization.
The formula \eqref{momentformel0} supports the interpretation of $W_f$ as a time-frequency distribution, since the right hand side is a normalized first-order frequency moment of the Wigner distribution, which delivers the center frequency of a narrowband signal.

The paper is organized as follows. After fixing some definitions and notation in Section \ref{prel} we introduce the framework of mean-square differentiable Gaussian proper stochastic processes in Section \ref{harmdiff}, which also contains a background for the special case of harmonizable processes. In Section \ref{ifwignerdeterm} we discuss and prove a precise version of the well-known formula \eqref{momentformel0} for the IF of a deterministic signal as a normalized first-order frequency moment of the Wigner distribution (for fixed time). In Section \ref{stochiftf} we briefly give a background on earlier work on the IF of stochastic processes and its relation to the Wigner spectrum. Then, in Section \ref{pdf}, we prove our main result, which is a formula for the pdf of the IF for fixed time and its relation to the Wigner spectrum. Finally we show by examples in Section \ref{cases} that a process may have infinite-variance IF for all time points or may have zero-variance IF for all time points.

\section{Preliminaries}\label{prel}

A probability space $(\Omega,\mathcal B,\p)$ consists of a space
$\Omega$, a $\sigma$-algebra $\mathcal B$ and a probability measure
$\p$. A random variable is a measurable function $X: \Omega \mapsto
\R^d$ (often denoted by a capital letter), and $X$
induces a probability measure on $\mathcal B(\R^d)$, the Borel
$\sigma$-algebra, defined by $P_X(A) = \p(X^{-1} (A))$, $A \in
\mathcal B(\R^d)$. If the probability measure $P_X$ is absolutely continuous with respect to Lebesgue measure we have $P_X(A) = \int_A p_X(x) dx$ where $p_X$ is the probability density function (pdf) of $X$. Sometimes, by abuse of notation, we will say that a probability measure that is a Dirac measure at a point $a \in \R^d$, denoted $\delta_a$, has pdf $\delta_a$.
We denote the expectation of a random variable $X$ by $\E X$, and by $L_0^2(\Omega)=L_0^2(\Omega,\mathcal B,\mathbb P)$ we understand the Hilbert space of second-order (finite variance) zero-mean complex-valued random variables. The space of second-order random variables with nonzero mean is denoted $L^2(\Omega)$.

Given a probability space $(\Omega,\mathcal B,\p)$, a continuous-time stochastic
process on $\R$ is defined as a family of complex-valued
$\B$-measurable functions $z_t(\omega)$ indexed by $t \in \R$. We
often suppress the variable $\omega \in \Omega$ and denote
$z(t)=z_t(\omega)$. Sometimes we write $z_t(\omega) =
z(t,\omega)$ to emphasize the fact that $z: \R \times \Omega \mapsto
\C$ is a function of two variables. For fixed $\omega$ the function
$t \mapsto z(t,\omega)$ is called a realization, trajectory or sample function.

For a space $U$, we denote by $\chi_A$ the indicator function of the set $A \subseteq U$, that is $\chi_A(x)=1$ if $x \in A$ and $\chi_A(x)=0$ if $x \notin A$.
The space of continuous functions on $\R$ is denoted $C(\R)$, the space of continuously differentiable functions on $\R$ is denoted $C^1(\R)$, and the space of continuous functions decaying at infinity is denoted $C_0(\R)$. This means that for any $\ep>0$ there is a compact set $K_\ep \subset \R$ such that $x \in K_\ep^c \Rightarrow |f(x)|<\ep$ \cite{Rudin1}, where $K_\ep^c = \R \setminus K_\ep$ denotes set complement. A derivative with respect to time is denoted $\dot{f}=df/dt$. This notation is also used for stochastic processes where we define the derivative in the mean-square sense. That is, a process $z(t)$ is mean-square differentiable at $t=t_0$ if there exists $\dot{z}(t_0) \in L^2(\Omega)$ such that
\begin{equation}\label{msder1}
\lim_{\ep \rightarrow 0} \E \left| \frac{z(t_0+\ep)-z(t_0)}{\ep} - \dot{z}(t_0)\right|^2 = 0.
\end{equation}
For a function $f(x,y)$ of two variables we write the partial derivative with respect to the first variable as $\partial_1 f(x,y) = \partial f(x,y)/\partial x$ and with respect to the second variable as $\partial_2 f(x,y) = \partial f(x,y)/\partial y$.
The normalization of the Fourier transform for functions $f\in L^1(\R)$ used in this paper is
$$
\mathscr{F} f(\xi) = \wh f(\xi) = \int_{\R} f(t) e^{-i t \xi} dx
$$
which gives the inverse Fourier transform
$$
f(t) = \mathscr{F}^{-1} \wh f(t) = \frac1{2 \pi} \int_{\R} \wh f(\xi) e^{i t \xi} d\xi.
$$
We denote by $\mathscr F L^1(\R)$ the space of functions $f$ with Fourier transform $\wh f \in L^1(\R)$.
For functions of several variables a partial Fourier transformation with respect to variable $j$ is denoted $\mathscr{F}_j$. The Wigner(--Ville) distribution \cite{Flandrin1,Folland1,Grochenig1} for $f \in L^2(\R)$ is defined and denoted by \eqref{wvd1}, where $\overline{f}$ denotes complex conjugate.

A tempered distribution $f$ belongs to the \emph{Sobolev space} of order $s \in \R$ \cite{Folland1}, denoted $f \in H^s(\R)$, provided its Fourier transform $\wh f$ is locally integrable and satisfies
\begin{equation}\nonumber
\int_{\R} (1+|\xi|^2)^{s} |\wh f(\xi)|^2 d\xi < \infty.
\end{equation}
The Sobolev scale is a smoothness scale since for $f \in H^s(\R)$ and $s$ large, a certain amount of asymptotic decay at infinity of the Fourier transform is required. This implies that $f$ will be differentiable to a degree that increases with $s$.

Finally we use the notations $\R_+ = [0,+\infty)$, $\R_- = (-\infty,0]$, the determinant of a square matrix $M$ is $\det(M)=|M|$ and the transpose of a vector $x$ is $x^T$.

\section{Mean square differentiable and harmonizable Gaussian stochastic processes}\label{harmdiff}

Let $z(t)=x(t)+i y(t)$ be a continuous-time, complex-valued, zero-mean Gaussian stochastic
process $z: \R \mapsto L_0^2(\Omega)$, not necessarily WSS.
The assumption that $z$ is Gaussian means the following \cite{Janson1}.
For any finite vector of time points $t=(t_j )_{j=1}^n$ the sampled process real-valued $2n$-vector $Z_t:=( x(t_1), \dots, x(t_n),y(t_1),\dots ,y(t_n) )^T$ has the pdf
\begin{equation}\nonumber
 p_{Z_t} (u) = (2\pi)^{-n} |M|^{-1/2} \exp\left(-\frac1{2} u^T M^{-1} u \right), \quad u \in \R^{2n},
\end{equation}
provided the covariance matrix $M=\E ( Z_t Z_t^T)$ is invertible.
A more general definition, which works also when $M$ is singular, is the requirement that the
sampled process $Z_t$ is a random variable with
characteristic function
$$
\phi_{Z_t}(u) = \E( \exp(i u^T Z_t) ) = \exp \left(-\frac1{2} u^T M u
\right), \quad u \in \R^{2n}.
$$
The covariance function of the process $z$ is denoted
$$
r_z(t,s) = \E( z(t) \overline{z(s)}), \quad t,s \in \R.
$$
We assume that $z$ is \emph{proper} \cite{Neeser1}, that is $\E( z(t) z(s) ) \equiv
0$, which implies
\begin{equation}\label{proper1}
r_x(t,s)=r_y(t,s), \quad r_{yx}(t,s)=-r_{yx}(s,t), \quad r_z(t,s) = 2 r_x(t,s) + 2 i
r_{yx}(t,s), \quad t,s \in \R,
\end{equation}
where $r_x(t,s)=\E(x(t) x(s))$, $r_y(t,s)=\E(y(t) y(s))$ and $r_{yx}(t,s)=\E(y(t) x(s))$.

We require that the process $z(t)$ be mean-square continuous and have a continuous mean-square derivative according to the following definition.

\begin{defn}\label{differentiable1}
A mean-square continuous process $z(t)$ is mean-square differentiable with continuous derivative $\dot{z}(t)$ if there exists $\dot{z}(t_0) \in L^2(\Omega)$ such that \eqref{msder1} is satisfied for all $t_0 \in \R$, and furthermore $\lim_{\ep \rightarrow 0} \E | \dot{z}(t_0+\ep) - \dot{z}(t_0) |^2 = 0$ for all $t_0 \in \R$.
\end{defn}
This definition guarantees that $\partial_1  r_z$, $\partial_1 \partial_2 r_z$ are continuous functions \cite{Cramer1,Gihman1,Loeve1} and
\begin{equation}\label{derivatakorr1}
\E ( \dot{z}(t) \overline{ z (s) } ) = \partial_1  r_z(t,s), \quad \E ( \dot{z}(t) \overline{ \dot{z}(s) } ) = \partial_1 \partial_2 r_z(t,s), \quad t,s \in \R.
\end{equation}

\begin{rem}
Note that in Definition \ref{differentiable1} we discuss a derivative process in the mean square sense only. That is, the definition does not imply that each realization is continuously differentiable with probability one. Conditions that are sufficient for continuously differentiable realizations with probability one are much more subtle and difficult \cite{Cramer1,Doob1,Miller1}.
\end{rem}

For some results we will assume that $z$ is strongly harmonizable (abbreviated in this paper as \emph{harmonizable}) \cite{Kakihara1,Loeve1,Rao1}, which
means that $r_z$ has a Fourier--Stieltjes representation
\begin{equation}\label{spektralintegral1}
r_z(t,s) = \iint_{\R^2} e^{i (t \xi - s \eta)} m_z(d\xi,d\eta),
\end{equation}
where $m_z$ is a measure of bounded variation on $\R^2$,
called the spectral measure. This assumption implies that the
process $z(t)$ has a Fourier transform representation
$$
z(t) = \int_{\R} e^{i t \xi} Z(d\xi),
$$
where the so-called spectral process $Z: \mathcal B(\R) \mapsto
L_0^2(\Omega)$ is a vector-valued measure of bounded semivariation \cite{Kakihara1}.

An important special case of harmonizable processes are the mean-square continuous WSS processes.
This means that there exists a continuous positive definite function $\rho_z$ such that $r_z(t,s)=\rho_z(t-s)$ \cite{Loeve1}. In the spectral domain the mean-square continuous WSS processes are characterized by $m_z(A,B) = \mu_z(A\cap B)$, $A,B \in \mathcal B(\R)$ for a non-negative bounded measure of one variable $\mu_z$.

For harmonizable processes the following requirement in the spectral domain is sufficient to guarantee that the process $z(t)$ is differentiable in the sense of Definition \ref{differentiable1}.
\begin{defn}\label{spekmomdef1}
A harmonizable process $z$ with spectral measure $m_z$ has spectral moments of order one if
\begin{equation}\label{spekmom1}
\iint_{\R^2} ( 1+|\xi|^2 )^{1/2} ( 1+|\eta|^2 )^{1/2} |m_z|(d\xi,d\eta) < \infty.
\end{equation}
\end{defn}

Here $|m_z|$ denotes the total variation measure of the complex measure $m_z$ \cite{Rudin1}.
Note that this definition implies that all four integrals
\begin{equation}\nonumber
\iint_{\R^2}  |m_z|(d\xi,d\eta), \quad \iint_{\R^2} |\xi| |\eta| |m_z|(d\xi,d\eta), \quad \iint_{\R^2} |\xi| |m_z|(d\xi,d\eta), \quad \iint_{\R^2} |\eta| |m_z|(d\xi,d\eta)
\end{equation}
are finite. Definition \ref{spekmomdef1} guarantees that we may take partial derivatives under the integral in \eqref{spektralintegral1} as
\begin{equation}\label{spektralderivata1}
\begin{aligned}
\partial_1 r_z(t,s) & = \iint_{\R^2} i \xi e^{i (t \xi - s \eta)} m_z(d\xi,d\eta), \\
\partial_2 r_z(t,s) & = \iint_{\R^2} (-i \eta) e^{i (t \xi - s \eta)} m_z(d\xi,d\eta), \\
\partial_1 \partial_2 r_z(t,s) & = \iint_{\R^2} \xi \eta e^{i (t \xi - s \eta)} m_z(d\xi,d\eta),
\end{aligned}
\end{equation}
due to Lebesgue's dominated convergence theorem \cite{Rudin1}. Moreover, the process satisfies Definition \ref{differentiable1} and \eqref{derivatakorr1} \cite{Loeve1}.

The Wigner distribution \eqref{wvd1} may be written $W_f=\mathscr{F}_2 (f \otimes \overline{f} \circ \kappa)$
where $\mathscr{F}_2$ denotes Fourier transformation in the second variable and $\kappa$ denotes the coordinate transformation
\begin{equation}\label{kappa}
\kappa(x,y)=(x+y/2,x-y/2) \quad \Longleftrightarrow \quad \kappa^{-1}(x,y)=((x+y)/2,x-y).
\end{equation}
The Wigner (-Ville) spectrum \cite{Flandrin1} of a harmonizable process $z(t)$ is
defined by
\begin{equation}\label{wvs1}
\wt W_z = \mathscr{F}_2 (r_z \circ \kappa).
\end{equation}
Since $r_z$ may not be an integrable function, the partial Fourier transformation in \eqref{wvs1} is in general defined with $r_z$ understood as a tempered distribution \cite{Folland1}. However, in the case when $r_z \in S_0(\R^2)$, which denotes Feichtinger's algebra \cite{Grochenig1}, we can write \eqref{wvs1} as the partial Fourier integral
\begin{equation}\label{wvs2}
\wt W_z (t,\xi) = \int_{\R} r_z (t+\tau/2,t-\tau/2) e^{-i \tau \xi} d\tau, \quad t \in \R.
\end{equation}
Under the same assumption (plus Gaussianity) we may interchange the order of integration and expectation so we also have
\begin{equation}\nonumber
\wt W_z (t,\xi) = \E\left( \int_{\R} z (t+\tau/2) \overline{z(t-\tau/2)} e^{-i \tau \xi} d\tau \right) = \E \left( W_z(t,\xi) \right), \quad t \in \R,
\end{equation}
that is, the Wigner spectrum is the expected value of the Wigner distribution of the stochastic process $z$ defined by \eqref{wvd1} \cite{Wahlberg1}.

Using \eqref{kappa} we may write the representation \eqref{spektralintegral1} as
\begin{equation}\label{spektralintegral2}
r_z \circ \kappa (t,s) = \iint_{\R^2} e^{i (s \xi + t \eta)} m_z
\circ \kappa (d\xi,d\eta)
\end{equation}
and thus by identification with \eqref{wvs1} it follows that
\begin{equation}\label{wvs2}
\wt W_z(t,d\xi) = 2 \pi \int_{\eta \in \R} e^{i t \eta} m_z \circ \kappa (d\xi,d\eta).
\end{equation}

\section{The IF and the Wigner distribution}\label{ifwignerdeterm}

The argument (or phase) of a complex number $z=x+iy \in \C \setminus \{ 0 \}$ is defined by $z=|z| e^{i \arg z}$ where we impose the restriction $-\pi < \arg z \leq \pi$. We have
\begin{equation}\label{argument1}
\arg z =
\left\{
  \begin{array}{ll}
    \arctan\left(\frac{y}{x}\right) , & x > 0,  \\
    \pi \sgn(y)+\arctan\left(\frac{y}{x}\right), & x<0, \quad y \neq 0, \\
    \frac\pi{2} \sgn(y), & x=0, \quad y \neq 0,  \\
    \pi, & x<0, \quad y=0, \\
    \mbox{undefined}, & x = y=0.
  \end{array}
\right.
\end{equation}
For $z \in U := \C \setminus \{ (-\infty,0] +i 0 \}$ we have $\arg z= \im ( \log z) = -i \log(z/|z|)$ where the principal branch of the logarithm function is understood, i.e. $-\pi<\im (\log z) < \pi$. The function $\log z$ is holomorphic from the domain $U$ onto $\R + i(-\pi,\pi) \subset \C$, and therefore $z \mapsto \arg z$ is a smooth function $U \mapsto (-\pi,\pi)$.

Let $f(t)=x(t) + i y(t)$ be a function $f: \R \mapsto \C$. If $f(t) \neq 0$ we reserve the notation 
$$
\varphi(t)=\arg f(t)
$$ 
for the phase function. For $f \in C(\R)$ the set $U_f := \{ t \in \R: f(t) \in U \} \subseteq \R$ is open and $\varphi = \arg \circ f: U_f \mapsto (-\pi,\pi)$ is continuous . If, moreover, $f \in C^1(\R)$ then it follows from above that $\varphi$ is differentiable on $U_f$, and since  $\frac{d}{dt} \arctan (t)=1/(1+t^2)$ the derivative is
\begin{equation}\label{fasderivata1}
\dot{\varphi}(t) = \frac{d}{dt} \arg f(t) = \frac{1}{1+\frac{y^2(t)}{x^2(t)}} \cdot \frac{x(t) \dot{y}(t) - \dot{x}(t) y(t)}{x^2(t)} = \frac{x(t) \dot{y}(t) - \dot{x}(t) y(t)}{x^2(t) + y^2(t)}
\end{equation}
for $t \in U_f$. In fact, for $\{ t \in \R: x(t)>0 \}$ and $\{ t: x(t)<0, \ y(t) \neq 0 \}$ this follows from \eqref{argument1}.
In the remaining case $\{ t \in \R: x(t)=0, \ y(t) \neq 0 \}$ we may use the following modified definition, equivalent to \eqref{argument1} in $\{ t \in \R: y(t) \neq 0 \}$,
$$
\arg z = \frac\pi{2} \sgn(y) - \arctan\left(\frac{x}{y}\right),
$$
which gives $\dot{\varphi}(t)=-\dot{x}(t)/y(t)$ for $\{ t \in \R: x(t)=0, \ y(t) \neq 0 \}$. This also leads to the expression on the right hand side of \eqref{fasderivata1}.

Clearly the formula \eqref{fasderivata1} can be extended from the domain $U_f$ to $\{ t: f(t) \neq 0 \}$ and $\dot{\varphi}$ is still continuous $\{ t: f(t) \neq 0 \} \mapsto \R$.
The \emph{instantaneous frequency} (IF) \cite{Cohen1,Cramer1,Flandrin1} of $f(t)$ is defined by \eqref{fasderivata1} as the derivative $\dot{\varphi}(t)$ with domain $\{ t: f(t) \neq 0 \}$. For an exponential function $e^{i \xi t}$ with frequency $\xi$ the IF is thus $\xi$ constantly, which means that the term instantaneous frequency is an extension of the concept of a constant (global) frequency.
For $\{ t: f(t)=0 \}$ it will turn out to be convenient to \emph{define} (by abuse of notation) $\dot{\varphi}(t)=+\infty$ so in summary we have for $f \in C^1(\R)$
\begin{equation}\label{if1}
\dot{\varphi}(t) =
\left\{
  \begin{array}{ll}
    \frac{x(t) \dot{y}(t) - \dot{x}(t) y(t)}{x^2(t) + y^2(t)} & \quad \mbox{if} \quad x^2(t) + y^2(t) > 0,  \\
    +\infty & \quad \mbox{if} \quad x^2(t) +y^2(t) =0.
  \end{array}
\right.
\end{equation}

If $f$ is real-valued and continuous then $f(t) \neq 0$ implies that $f$ has constant sign, that is $\varphi(t)=0$ or $\varphi(t)=\pi$, in a neighborhood of $t$, and hence $\dot{\varphi}(t)$ is well-defined and equals zero in this neighborhood of $t$. The derivative of the phase function of real-valued signals is thus not very interesting. However, the Hilbert transform \cite{Cramer1,Papoulis1} of a real-valued signal gives rise to a complex-valued signal with a nonzero IF. Since this transformation from a real-valued signal to a so-called \emph{analytic} signal transforms $\cos(\xi t)$ into $e^{i \xi t}$, it gives a natural definition of the IF of a real-valued signal, commonly used in the literature.

There is a connection between the IF of a sufficiently smooth and decaying function and the Wigner distribution \cite{Cohen1,Flandrin1}. The heuristic version of this result is well known \cite{Cohen1} but here we prove a more precise statement.

\begin{prop}\label{ifwignerprop1}
Suppose $\ep>0$, $f \in H^{\frac{3}{2}+\ep}(\R)$, $f(t)=x(t) + i y(t)$ and $\dot{\varphi}(t)$ is defined by \eqref{if1}. Then for any $t \in \R$ such that $f(t) \neq 0$ we have
\begin{equation}\label{ifwigner1}
\dot{\varphi}(t) = \frac{\int_{\R} \xi W_f(t,\xi) d\xi }{ \int_{\R} W_f(t,\xi) d\xi }.
\end{equation}
\end{prop}

\begin{proof}
The assumption is
\begin{equation}\label{sobolev1}
\int_{\R} (1+|\xi|^2)^{\frac{3}{2}+\ep} |\wh f(\xi)|^2 d\xi < \infty.
\end{equation}
Let $t \in \R$ be fixed and arbitrary such that $f(t) \neq 0$.
Define the function $g(\tau) := f(t+ \tau/2) \overline{f(t-\tau/2)}$. The requirement \eqref{sobolev1} implies that $\wh f \in L^2(\R)$ which means that $f \in L^2(\R)$ by Plancherel's theorem. Thus by the Cauchy--Schwarz inequality we have
\begin{equation}\label{finvsuff1}
g = f(t+ \cdot/2) \overline{f(t-\cdot/2)} \in L^1(\R).
\end{equation}
The Cauchy--Schwarz inequality and \eqref{sobolev1} give
\begin{equation}\label{sobolev2}
\begin{aligned}
\int_{\R} (1+|\xi|^2)^{1/2} |\wh f(\xi)| d\xi & = \int_{\R} (1+|\xi|^2)^{\frac1{2} + \frac1{4} + \frac{\ep}{2} - \frac1{4} - \frac{\ep}{2}} |\wh f(\xi)| d\xi  \\
& \leq \left( \int_{\R} (1+|\xi|^2)^{\frac3{2}+\ep} |\wh f(\xi)|^2 d\xi \right)^{1/2} \left( \int_{\R} (1+|\xi|^2)^{-\frac{1}{2} - \ep } d \xi \right)^{1/2} < \infty.
\end{aligned}
\end{equation}
Hence $\wh f \in L^1(\R)$ and since $\mathscr F (f(t+ \cdot/2))(\xi) = 2 e^{i 2 \xi t} \wh f(2 \xi)$ we have $\mathscr F (f(t+ \cdot/2)) \in L^1(\R)$. Thus $\mathscr F (\overline{f(t- \cdot/2)})=\overline{\mathscr F (f(t+ \cdot/2))}$ implies that
\begin{equation}\label{finvsuff2}
\mathscr F g = \mathscr F (f(t+ \cdot/2) \overline{f(t-\cdot/2)}) = \frac1{2\pi} \mathscr F (f(t+ \cdot/2)) * \overline{\mathscr F (f(t+ \cdot/2))} \in L^1(\R).
\end{equation}
Moreover, $\wh f \in L^1(\R)$ implies that $f \in C_0(\R)$ according to the Riemann--Lebesgue lemma. Also, since $\mathscr F \dot{f}(\xi)=i \xi \wh f(\xi)$ (where $f$ is considered a tempered distribution), and $\xi \wh f(\xi) \in L^1(\R)$ according to \eqref{sobolev2}, we have $\dot{f} \in C_0(\R)$ again by the Riemann--Lebesgue lemma. Thus by \eqref{finvsuff1} and \eqref{finvsuff2} we have $g \in L^1(\R) \cap \mathscr F L^1(\R) \cap C_0(\R)$, which means that Fourier's inversion formula holds pointwise \cite{Rudin1} for the function $g$: We have $g(\tau) = (\mathscr F^{-1} \wh g) (\tau)$ for all $\tau \in \R$. In particular
\begin{equation}\label{namnare1}
2\pi |f(t)|^2 = 2 \pi g(0) = \int_\R \wh g(\xi) d\xi = \int_\R W_f(t,\xi) d \xi,
\end{equation}
where the last equality uses the definition \eqref{wvd1}.

Concerning the numerator in \eqref{ifwigner1} we use integration by parts and the just proven fact that $f$ vanishes at infinity to obtain
\begin{equation}\label{intbypart1}
\begin{aligned}
\xi W_f(t,\xi) & = \int_{\R} f(t+\tau/2) \overline{f(t-\tau/2)} i \frac{d}{d \tau} (e^{-i \tau \xi}) d \tau \\
& = -i \int_{\R} e^{-i \tau \xi} \frac{d}{d \tau} \left( f(t+\tau/2) \overline{f(t-\tau/2)} \right) d \tau.
\end{aligned}
\end{equation}
Let us study the function $g_1(\tau)= \dot{f}(t+\tau/2) \overline{f(t-\tau/2)}$ and
$$
h(\tau) = \frac{d}{d \tau} \left( f(t+\tau/2) \overline{f(t-\tau/2)} \right)
= \frac{1}{2} (g_1(\tau)-\overline{g_1(-\tau)}).
$$
The assumption \eqref{sobolev1} implies that $\xi \wh f(\xi) \in L^2(\R)$ and thus $\mathscr F \dot{f} \in L^2(\R)$, and $\dot{f} \in L^2(\R)$ by Plancherel's theorem. Since we already know that $f \in L^2(\R)$ the Cauchy--Schwarz inequality gives $g_1 \in L^1(\R)$.

 From above we know that $\mathscr F (\overline{f(t- \cdot/2)})=\overline{\mathscr F (f(t+ \cdot/2))} = 2 e^{- i 2 \xi t} \overline{\wh f(2 \xi)} \in L^1(\R)$. Likewise we have $\mathscr F(\dot{f}(t+ \cdot/2))(\xi) = 2 e^{i 2 \xi t} \widehat{\dot{f}}(2 \xi) = 2 i e^{i 2 \xi t} 2 \xi \wh f(2 \xi)$. This function belongs to $L^1(\R)$ because of \eqref{sobolev2}. Thus
$$
\wh g_1 = \frac1{2\pi} \mathscr F (\dot{f}(t + \cdot/2)) * \mathscr F (\overline{f(t- \cdot/2)}) \in L^1(\R).
$$
Hence we have proved that $g_1 \in C_0(\R) \cap L^1(\R) \cap \mathscr F L^1(\R)$ which means that Fourier's inversion formula holds for $g_1$ and hence also for $h$. Denoting $f(t)=x(t)+i y(t)$, integration of \eqref{intbypart1} gives
\begin{equation}\label{taljare1}
\begin{aligned}
\int_\R \xi W_f(t,\xi) d \xi & = - i \int_\R \mathscr F h(\xi) d \xi = - 2 \pi i h(0) = - \pi i \left( \dot{f}(t) \overline{f(t)} - f(t) \overline{\dot{f}(t)} \right) \\
& = 2 \pi ( x(t) \dot{y}(t) - \dot{x}(t) y(t) ).
\end{aligned}
\end{equation}
Finally \eqref{ifwigner1} follows from a combination of \eqref{if1}, \eqref{namnare1} and \eqref{taljare1}.
\end{proof}

\section{The IF for stochastic processes and the Wigner spectrum}\label{stochiftf}

One of our goals in this paper is to generalize Proposition \ref{ifwignerprop1} from deterministic functions to certain stochastic processes defined on $\R$, denoted $z(t)=x(t)+i y(t) = |z(t)| e^{i \varphi(t)}$. More precisely we would like to prove the formula
\begin{equation}\label{vvformel1}
\E(\dot{\varphi}(t)) = \frac{\int_{\xi \in \R} \xi \wt W_z(t,d\xi) }{ \int_{\xi \in \R} \wt W_z(t,d\xi)}.
\end{equation}
In fact, we will compute the pdf of the random variable $\dot{\varphi}(t)$, for fixed $t \in \R$, and then as a consequence derive the formula above.
This problem has been studied by Miller \cite{Miller1} and Broman \cite{Broman1} for WSS Gaussian proper stochastic processes. Miller and Broman independently derived the probability density function for $\dot{\varphi}(t)$ for fixed $t$, using either the hypothesis that the process is proper (in \cite{Miller1}) or the more restrictive hypothesis that the signal is analytic \cite{Papoulis1} (in \cite{Broman1}). (See also \cite{Strom1}.) Our aim is to generalize their results from WSS to Gaussian proper \emph{nonstationary} processes that have mean-square continuous derivative. As a special case we will study certain harmonizable processes and prove the identity \eqref{vvformel1} for them.

For processes more general than WSS it has been customary, in parts of the literature, to define the instantaneous frequency as \cite{Flandrin1,Martin1}
\begin{equation}\label{fuskdefinition1}
\psi(t) := \frac{x(t) \dot{y}(t) - \dot{x}(t) y(t)}{\E |z(t)|^2}, \quad \E |z(t)|^2 > 0.
\end{equation}
This definition replaces the random variable in the denominator of \eqref{if1} by its expected value.
It gives a generalization of the formula \eqref{ifwigner1} of Proposition \ref{ifwignerprop1}, with the IF $\dot{\varphi}(t)$ replaced by $\E(\psi(t))$ and $W_f$ replaced by $\wt W_f$, as follows.

\begin{prop}\label{ifwignerprop2}
Suppose $z(t)$ is a proper harmonizable process such that the spectral measure satisfies \eqref{spekmom1}.
Then for any $t \in \R$ such that $\E |z(t)|^2 > 0$ we have
\begin{equation}\label{ifwigner2}
\E(\psi(t)) = \frac{\int_{\xi \in \R} \xi \wt W_z(t,d\xi) }{ \int_{\xi \in \R} \wt W_z(t,d\xi)}.
\end{equation}
\end{prop}

\begin{proof}
First we note that \eqref{spektralintegral1}, \eqref{kappa} and \eqref{wvs2} give
\begin{equation}\label{namnare2}
\E|z(t)|^2 = r_z(t,t) = \iint_{\R^2} e^{i t (\xi-\eta)} m_z(d \xi,d\eta) = \iint_{\R^2} e^{i t \eta} m_z \circ \kappa (d \xi,d\eta) = \frac1{2 \pi} \int_{\xi \in \R} \wt W_z(t,d \xi).
\end{equation}
Next, since $\E(x(t) \dot{y}(t)) = \partial_1 r_{yx}(t,t)$ and $\E(\dot{x}(t) y(t)) = \partial_2 r_{yx}(t,t)$, we obtain using \eqref{proper1}, \eqref{spektralderivata1} and \eqref{wvs2}
\begin{equation}\label{taljare2}
\begin{aligned}
\E \left( x(t) \dot{y}(t) - \dot{x}(t) y(t) \right) & = \partial_1 r_{yx}(t,t) - \partial_2 r_{yx}(t,t) \\
& = \frac1{2i} \left( \partial_1 r_z (t,t) - \partial_2 r_z(t,t)  \right) \\
& = \frac1{2i} \left( \iint_{\R^2} i (\xi + \eta) e^{i t (\xi-\eta)} m_z (d\xi,d\eta) \right) \\
& = \iint_{\R^2} \xi e^{i t \eta} m_z \circ \kappa (d\xi,d\eta) \\
& = \frac1{2 \pi} \int_{\xi \in \R} \xi \wt W_z (t, d \xi).
\end{aligned}
\end{equation}
Finally \eqref{ifwigner2} follows from a combination of \eqref{fuskdefinition1}, \eqref{namnare2} and \eqref{taljare2}.
\end{proof}

The definition \eqref{fuskdefinition1} is however not the natural definition of the IF stochastic process. Instead we define the IF by \eqref{if1} where $z(t)=x(t)+i y(t) = |z(t)| e^{i \varphi(t)}$ is a zero-mean, complex-valued, proper, Gaussian process that is differentiable in the sense of Definition \ref{differentiable1}. Thus $\dot{\varphi}(t)$ is a $\R \cup \{ +\infty \}$-valued stochastic process defined on $\R$.

\begin{rem}\label{singlefunction1}
A trivial case of a second-order zero-mean stochastic process consists of a random variable $X \in L_0^2(\Omega)$ times a function $f \in C^1(\R)$, that is $z(t)=X f(t)$. For such a process we have $z(t,\omega)=|X(\omega)| |f(t)| e^{i (\arg X(\omega) + \varphi(t))}$, where $f(t)=|f(t)| e^{i \varphi(t)}$, provided $f(t) \neq 0$ and $X(\omega) \neq 0$. It follows that the IF of $z(t,\omega)$ is $\dot{\varphi}(t)$ provided that $f(t) \neq 0$ and $X(\omega) \neq 0$.
That is, the IF of $z(t)$ is $\dot{\varphi}(t) \chi_{\{\omega: X(\omega) \neq 0\}}(\omega) + \infty \chi_{\{\omega: X(\omega) = 0\}}(\omega)$. This means that the IF is essentially the deterministic IF of $f$ and is only stochastic in the sense that for certain $\omega \in \Omega$ it is $+\infty$ for all $t \in \R$, and for the remaining $\omega$ it does not depend on $\omega$. This is true also when $X$ is non-Gaussian and improper, that is $\E X^2 \neq 0$.
\end{rem}

\section{The probability density function for the stochastic IF for fixed time}\label{pdf}

In this section we will derive the pdf of the IF stochastic process $\dot{\varphi}(t)$, defined by \eqref{if1} for a mean-square differentiable stochastic process $z(t) = x(t)+i y(t) = |z(t)| e^{i \varphi(t)}$, for a fixed arbitrary $t \in \R$. As a consequence we will obtain the formula \eqref{vvformel1} for $\E (\dot{\varphi}(t) )$ provided the process is harmonizable.

Let $z(t)=x(t)+i y(t)$ be a proper Gaussian stochastic process which is differentiable according to Definition \ref{differentiable1}.
Fix $t \in \R$ and define the random column four-vector $X=(x(t), \dot{y}(t), y(t),
\dot{x}(t))^T$. Then $X$ is zero-mean and Gaussian, since Gaussianity is
preserved under mean square limits \cite{Janson1}. Because $z(t)$ is proper,
\eqref{proper1} holds, and thus $\E x(t) y(t)= \E \dot{y}(t) \dot{x}(t) = \pd_1
\pd_2 r_{yx}(t,t)=0$ and $\E \dot{y}(t) y(t) = \pd_1 r_y(t,t)=\pd_1
r_x(t,t)=\E x(t) \dot{x}(t)$. Hence $X$ has covariance matrix $M=\E X X^T \in \R^{4 \times 4}$ with the structure
\begin{eqnarray}\label{Mdef}
M =\left(
  \begin{array}{rrrr}
    a & b & 0 & c \\
    b & d & c & 0 \\
    0 & c & a & -b \\
    c & 0 & -b & d
  \end{array}
  \right),
\end{eqnarray}
and parameter values
\begin{equation}\label{parametrar1}
\begin{aligned}
a & =\E x(t)^2=r_x(t,t), \quad b=\E x(t)
\dot{y}(t) = \pd_1 r_{yx}(t,t), \\
c & = \E \dot{y}(t) y(t) = \pd_1 r_{x}(t,t), \quad d=\E \dot{y}(t) \dot{y}(t) = \pd_1 \pd_2 r_{x}(t,t).
\end{aligned}
\end{equation}
We have $|M| = (a d -c^2 - b^2)^2$ and
\begin{equation}\label{parameterolikhet1}
a d -c^2 - b^2 \geq 0,
\end{equation}
which follows from taking the determinant of the upper left $3
\times 3$ submatrix of the nonnegative definite matrix $M$.

Our main technical result concerns real Gaussian zero-mean four-vectors with covariance matrix \eqref{Mdef}. In the proof we will need a small lemma that may be considered ``folkloristic'' in probability theory and not necessary to prove. Nevertheless we prove it for completeness and clarity.

\begin{lem}\label{nollmlem1}
Let $X: \Omega \mapsto \R^d$ be a random variable. Suppose that $N
\in \mathcal B(\R^d)$, $x_0 \in \R^d$ are two disjoint null sets for
$P_X$, i.e. $P_X(N)=P_X(\{ x_0 \})=0$, $x_0 \notin N$. If we define
$\wt X: \Omega \mapsto \R^d \setminus N$ by
\begin{equation}\label{xtilde1}
\wt X(\omega) = \left\{ \begin{array}{ll}
             X(\omega), & \omega \in \Omega \setminus X^{-1}(N),  \\
             x_0, & \omega \in X^{-1}(N),
           \end{array} \right.
\end{equation}
then $P_{\wt X}(A) = P_X(A)$ for all $A \in \B(\R^d)$.
\end{lem}
\begin{proof}
Let $A \in \B(\R^d)$. We write $A = A_1 \cup A_2 \cup A_3$ as a
pairwise disjoint union with $A_1 = A \cap N$, $A_2 = (A \cap \{x_0
\}) \setminus N$ and $A_3 = A \setminus (N \cup \{ x_0 \})$. Since
$\wt X$ takes values in $\R^d  \setminus N$ we have $P_{\wt X}(A_1)
\leq P_{\wt X}(N)=0$, and moreover $P_{\wt X}(A_2) \leq P_{\wt X}
(\{ x_0 \}) = \p(X^{-1}(N) \cup X^{-1}(x_0)) \leq P_X(N)+P_X(\{ x_0
\}) = 0$. Thus $P_{\wt X} (A) = P_{\wt X}(A_1) + P_{\wt X}(A_2) +
P_{\wt X}(A_3) = P_{\wt X}(A \setminus (N \cup \{x_0 \}) ) = P_X(A
\setminus (N \cup \{x_0 \}) )$, where the final equality follows
from $\wt X(\omega) \in A \setminus (N \cup \{x_0 \})
\Leftrightarrow \omega \in X^{-1}(A) \setminus (X^{-1}(N) \cup
X^{-1}(x_0)) \Leftrightarrow X(\omega) \in A \setminus (N \cup \{x_0
\})$. This finally gives $P_X(A) = P_X(A \cap(N \cup \{ x_0 \}) ) +
P_X(A \setminus (N \cup \{x_0\}) ) = P_X(A \setminus (N \cup
\{x_0\}) ) = P_{\wt X}(A)$ because $P_X(A \cap(N \cup \{ x_0 \}) )
\leq P_X(N) + P_X(\{ x_0 \}) = 0$.
\end{proof}

Next follows the crucial result about Gaussian four-vectors with covariance matrix $M$ according to \eqref{Mdef}.

\begin{prop}\label{grundprop1}
Let $X=(X_1,X_2,X_3,X_4)$ be a zero-mean Gaussian real-valued vector
with covariance matrix $M$ defined in \eqref{Mdef}. Define the $\R \cup \{ +\infty \}$-valued random variable
\begin{equation}\label{ifvar1}
Y = \left\{ \begin{array}{ll}
             \frac{X_1 X_2 - X_3 X_4}{X_1^2 + X_3^2}, & \mbox{if} \quad X_1^2 + X_3^2 > 0,  \\
             +\infty & \mbox{if} \quad X_1^2 + X_3^2 = 0.
           \end{array} \right.
\end{equation}
If $a d -c^2 - b^2>0$ then $Y$ has probability density function
\begin{equation}\label{pdf37}
p_Y(y) = \frac{a}{2} |M|^{1/2} \left(  (a y -b)^2 + |M|^{1/2}
\right)^{-3/2}.
\end{equation}
Consequently $Y$ has infinite variance and mean
\begin{equation}
\E Y = \frac{b}{a}.
\end{equation}
If $a d -c^2 - b^2=0$ we have the two subcases:

(i) If $a>0$ then $Y=b/a$ a.s.

(ii) If $a=0$ then $Y = + \infty$ a.s.

\end{prop}

\begin{proof}
Suppose first that $a d -c^2 - b^2>0$ which means that $M$ is invertible
with inverse
\begin{eqnarray}\label{Minv}
M^{-1}
 =|M|^{-1/2} \left(
  \begin{array}{rrrr}
    d & -b & 0 & -c \\
    -b & a & -c & 0 \\
    0 & -c & d & b \\
    -c & 0 & b & a
  \end{array}
  \right).
\end{eqnarray}
Furthermore $a d -c^2 - b^2>0$ implies $a>0$.
Define the smooth function $f: \R_+ \times (-\pi,\pi] \times \R
\times \R \longmapsto \R^4$, $f(y)=(f_1(y), f_2(y), f_3(y), f_4(y))$
by
\begin{equation}\label{f}
\begin{split}
f_1(y) & = y_1 \cos y_2, \\
f_2(y) & = y_3 \sin y_2 + y_1 y_4 \cos y_2, \\
f_3(y) & = y_1 \sin y_2, \\
f_4(y) & = y_3 \cos y_2 - y_1 y_4 \sin y_2.
\end{split}
\end{equation}
The Jacobian of $f$ is
\begin{eqnarray}\nonumber
Df(y)
 =\left(
  \begin{array}{cccc}
    \cos y_2 & -y_1 \sin y_2 & 0 & 0 \\
    y_4 \cos y_2 & y_3 \cos y_2 - y_1 y_4 \sin y_2 & \sin y_2 & y_1 \cos y_2 \\
    \sin y_2 & y_1 \cos y_2 & 0 & 0 \\
    -y_4 \sin y_2 & -y_3 \sin y_2 - y_1 y_4 \cos y_2 & \cos y_2 & -y_1 \sin y_2
  \end{array}
  \right),
\end{eqnarray}
whose determinant is $\det Df(y) = y_1^2$. Note that $(y_1,y_2) \mapsto
(f_1(y),f_3(y))$ is the polar-to-rectangular coordinate
transformation on $\R^2$.

We will use $f$ as a coordinate transformation, and then we will need $f$ to be a bijection with a differentiable inverse. Since $f(0,y_2,0,y_4) = 0$ for any $y_2 \in (-\pi,\pi]$ and any $y_4 \in \R$, the function $f$ is not injective on the domain $\R_+ \times (-\pi,\pi] \times \R
\times \R$. Therefore we need to restrict the domain of $f$. Define
\begin{equation}\label{Nprimdef}
N'= ( \{0\} \times (-\pi,\pi] \times \R \times \R ) \bigcup ( \R_+ \times \{ \pi \} \times \R \times \R ) \subset \R_+ \times (-\pi,\pi] \times \R \times \R
\end{equation}
and
\begin{equation}\label{Ndef}
N= \{x \in \R^4: x_1 \leq 0, \ x_3=0 \} = \R_- \times \R \times \{0\} \times \R \subset \R^4.
\end{equation}
Then it can be verified that $f(N')=N$ and the restriction
\begin{equation}\label{frestriction1}
f: \R_+ \times (-\pi,\pi] \times \R \times \R \setminus N' \longmapsto \R^4 \setminus N
\end{equation}
of $f$ to the open set $\R_+ \times (-\pi,\pi] \times \R \times \R \setminus N'$ is surjective and injective.
Its inverse is $f^{-1}=g=(g_1,g_2,g_3,g_4)$ where
\begin{equation}\label{g}
\begin{split}
g_1(x) & = (x_1^2 + x_3^2)^{1/2} , \\
g_2(x) & = \arg(x_1 + i x_3), \\
g_3(x) & = \frac{x_2 x_3 + x_1 x_4}{(x_1^2 + x_3^2)^{1/2}}, \\
g_4(x) & = \frac{x_1 x_2 - x_3 x_4}{x_1^2 + x_3^2}.
\end{split}
\end{equation}
Since $x \notin N \Rightarrow x_1^2 + x_3^2>0$ it is clear that $g_1$, $g_3$, and $g_4$ are differentiable on $\R^4 \setminus N$. Since the nonpositive $x_1$-axis in the $(x_1,x_3)$-plane, in conjunction with any $(x_2,x_4) \in \R^2$, does not belong to $\R^4 \setminus N$, $g_2$ is differentiable on $\R^4 \setminus N$. Hence the restriction of $f$ \eqref{frestriction1} is differentiable and has a differentiable inverse.

The sets $N' \subset \R_+ \times (-\pi,\pi] \times \R \times \R$ and $N \subset \R^4$ are null sets with respect to Lebesgue measure, and $P_X(N)=0$. Let $x_0 \in \R^4 \setminus N$ be fixed arbitrary and define the
random variable $\wt X: \Omega \mapsto \R^4 \setminus N$ by
\begin{equation}\label{modvar1}
\wt X(\omega) = \left\{ \begin{array}{ll}
             X(\omega), & \omega \in \Omega \setminus X^{-1}(N),  \\
             x_0 & \omega \in X^{-1}(N).
           \end{array} \right.
\end{equation}
It follows from Lemma \ref{nollmlem1} that $\wt X$ and $X$ induce
identical probability measures, i.e. $P_{\wt X}(A)=P_X(A)$ for all
$A \in \B(\R^4)$, which means that
\begin{equation}\label{tildepdf1}
P_{\wt X}(A) = (2 \pi)^{-2} |M|^{-1/2} \int_{A} \exp\left(-\frac1{2}x^T M^{-1}
x\right) dx, \quad A \in \B(\R^4),
\end{equation}
because $X$ is a zero-mean Gaussian with covariance matrix $M$.

Define the random variable $Z: \Omega \mapsto \R^4$ by $Z=f^{-1}(\wt
X)$, which is well defined because $\wt X(\Omega) = \R^4 \setminus N$ and the restriction \eqref{frestriction1} is bijective. If we write $Z=(Z_1,Z_2,Z_3,Z_4)$ and $x_0=(x_{0,1},x_{0,2},x_{0,3},x_{0,4})$ then according to \eqref{ifvar1}, \eqref{g}
and \eqref{modvar1} we have
\begin{equation}\label{ifvar2}
Z_4(\omega) = \left\{ \begin{array}{ll}
             Y(\omega), & \omega \in \Omega \setminus X^{-1}(N),  \\
             \frac{x_{0,1} x_{0,2} - x_{0,3} x_{0,4}}{x_{0,1}^2 + x_{0,3}^2}, & \omega \in X^{-1}(N).
           \end{array} \right.
\end{equation}
We claim that the random variables $Y$ and $Z_4$ have identical probability measures, that is
\begin{equation}\label{probdisteq1}
P_Y(A) = P_{Z_4}(A), \quad A \in \mathcal B(\R).
\end{equation}
In fact, let $A \in \B(\R)$. We decompose $Y^{-1}(A) = (Y^{-1}(A) \cap X^{-1}(N))
\bigcup (Y^{-1}(A) \setminus X^{-1}(N))$. Since $\p( Y^{-1}(A)
\cap X^{-1}(N) ) \leq \p(X^{-1}(N))=P_X(N)=0$ and $Y^{-1}(A)
\setminus X^{-1}(N) = Z_4^{-1}(A) \setminus X^{-1}(N)$ according to \eqref{ifvar2}, we obtain
$P_Y(A) = \p(Y^{-1}(A)) = \p(Y^{-1}(A) \setminus X^{-1}(N)) =
\p(Z_4^{-1}(A) \setminus X^{-1}(N)) = \p(Z_4^{-1}(A)) = P_{Z_4}(A)$,
proving \eqref{probdisteq1}.

Denote the probability density functions for $X$, $\wt X$ and $Z$ by
$p_X$, $p_{\wt X}$ and $p_Z$, respectively. Then \eqref{tildepdf1}
says that $p_{\wt X}(x) = p_X (x) = (2 \pi)^{-2} |M|^{-1/2} \exp\left(-\frac1{2}x^T
M^{-1} x\right)$, $x \in \R^4$. For an arbitrary Borel set $A \in \mathcal B(\R_+
\times (-\pi,\pi] \times \R \times \R)$ we have
\begin{equation}\nonumber
\begin{aligned}
P_Z(A) & = \p(Z \in A) = \p( Z \in A \setminus(X \in N) ) = \p( \wt X \in f(A) \setminus(X \in N) ) \\
& = \p( X \in f(A) \setminus N ) = \p( X \in f(A \setminus N') ) \\
& = \int_{f(A \setminus N')} p_X (x) dx = \int_{A \setminus N'} p_X \circ f(y) | \det Df(y) | dy \\
& = \int_A p_X \circ f(y) | \det Df(y) | dy \\
& = \int_A p_X \circ f(y) | \det Df(y) | \ \chi_{[0,+\infty)}(y_1) \chi_{(-\pi,\pi]}(y_2) \ dy.
\end{aligned}
\end{equation}
In fact, the seventh equality above is the formula for changing variables in integrals
\cite[Theorem 5.8]{Fleming1}. To justify its use, we need the fact that \eqref{frestriction1} is differentiable and has a differentiable inverse, which has been proved above. On inserting \eqref{Minv} and \eqref{f} we obtain
\begin{equation}\label{pdf1}
\begin{split}
p_Z(z) & = p_X ( f(z) ) | \det Df(z) | \chi_{[0,+\infty)}(z_1) \chi_{(-\pi,\pi]}(z_2) \\
& = (2\pi)^{-2} |M|^{-\frac1{2}} z_1^2 \exp\left(-
\frac1{2} f(z)^T M^{-1} f(z) \right) \chi_{[0,+\infty)}(z_1) \chi_{(-\pi,\pi]}(z_2) \\
& = (2\pi)^{-2} |M|^{-\frac1{2}} z_1^2 \exp\left(- |M|^{-\frac1{2}} \left(
z_1^2(d+a z_4^2 -2 b z_4) +a z_3^2 -2c z_1 z_3 \right)/2 \right) \chi_{[0,+\infty)}(z_1) \chi_{(-\pi,\pi]}(z_2)
\end{split}
\end{equation}
after some computations. Using $\int_0^\infty x^2 \exp(-s x^2/2 ) dx
= \sqrt{\pi/2} s^{-3/2}$, $s > 0$, and $|M| = (a d -c^2 - b^2)^2$,
we obtain the marginal probability density for $Z_4$
\begin{equation}\nonumber
\begin{split}
p_{Z_4}(z_4) & = \int_0^\infty \int_{-\pi}^{\pi} \int_{-\infty}^\infty
p_Z ( z_1,z_2,z_3,z_4 ) d z_1 d z_2 d
z_3  \\
& = (2\pi)^{-1} |M|^{-1/2} \int_0^\infty z_1^2 \exp\left(-
|M|^{-1/2} \left( z_1^2(d+a z_4^2 -2 b z_4)/2 \right)
\right) \\
& \quad \times \left( \int_{-\infty}^\infty \exp\left(- |M|^{-1/2}
\left( a z_3^2 -2c z_1 z_3 \right)/2 \right) d
z_3 \right) d z_1 \\
& = (2\pi a)^{-1/2} |M|^{-1/4} \int_0^\infty z_1^2 \exp\left(-
|M|^{-1/2} z_1^2(d+a z_4^2 -2 b z_4 - c^2/a)/2
\right) d z_1 \\
& = \frac{a}{2} |M|^{1/2} \left(  (a z_4 -b)^2 + |M|^{1/2}
\right)^{-3/2}.
\end{split}
\end{equation}
The earlier noted observation \eqref{probdisteq1} now proves \eqref{pdf37}. Finally, $\frac{d}{dx} ( x(x^2+s)^{-1/2}) = s(x^2+s)^{-3/2}$ gives
$$
\int_{-\infty}^\infty \frac{dx}{(x^2 + s)^{3/2}} = \frac2{s}
\int_0^\infty \frac{d}{dx} ( x(x^2+s)^{-1/2}) dx = \frac{2}{s},
\quad s > 0,
$$
so we have
\begin{equation}\nonumber
\begin{split}
\E Z_4 & = \frac{a}{2} |M|^{1/2} \int_{-\infty}^\infty \frac{z
dz}{\left(
(a z -b)^2 + |M|^{1/2} \right)^{3/2}} \\
& = \frac{|M|^{1/2} b}{2a}  \int_{-\infty}^\infty \frac{dz}{\left(
z^2 + |M|^{1/2} \right)^{3/2}} \\
& = \frac{b}{a}.
\end{split}
\end{equation}
Note that $p_{Z_4}(y)$ behaves like $C y^{-3}$ for large $y$. Therefore
$$
\int_\R (y-\E Y_4)^2 p_{Z_4}(y)dy=+\infty
$$
which means that the variance of $Z_4$, and therefore also that of $Y$, is infinite.
This proves the proposition in the case $a d -c^2 - b^2>0$.

It remains to consider the case when $a d -c^2 - b^2=0$, i.e. $M$ is
not invertible. If $a=0$ then $X_1=0$ a.s. and $X_3=0$ a.s.
$\Rightarrow Y = +\infty$ a.s. This proves case (ii). Assume henceforth that $a>0$. If
$a>0$ and $d=0$ then $X_2=0$ a.s., $X_4=0$ a.s. and $X_1^2+X_3^2>0$
a.s. $\Rightarrow Y = 0$ a.s. Since $b^2+c^2=ad=0$ we have $b=0$ so case (i) is proved if $d=0$.

The rest of the proof is devoted to the case (i) with $a d -c^2 - b^2=0$, $a>0$ and $d>0$. We
will employ a regularization technique. The characteristic
polynomial of $M$ is $\det(\lambda I - M) = \lambda^2(\lambda-(a+d))^2$.
Let
$$
U_{a+d} = \mathcal N(M - (a+d) I) \subset \R^4
$$
denote the two-dimensional eigenspace for $M$ corresponding to the nonzero
eigenvalue $\lambda = a+d$. The probability measure $P_X$ is
a Gaussian which is supported on $U_{a+d}$ and non-degenerate on this two-dimensional subspace.
With $N$ defined by \eqref{Ndef}, it can be verified that $N \cap U_{a+d} \subset W \subset \R^4$ where $W$ is a linear subspace with $\dim W=1$. This means that
$P_X(N)=0$. If $\wt X$ is defined by \eqref{modvar1} for some
$x_0 \in \R^4 \setminus N$, then Lemma \ref{nollmlem1} gives $P_X = P_{\wt X}$.

Let $n>0$ be an integer and define $X^{(n)}=(X_1,X_2 + X_2'/\sqrt{n},X_3, X_4 + X_4'/\sqrt{n})$ where $X_2'$, $X_4'$ are Gaussian zero-mean, unit-variance random variables, pairwise independent of each other and of $X_1$, $X_2$, $X_3$ and $X_4$. Then $X^{(n)}: \Omega \mapsto \R^4$ is a
Gaussian random variable with covariance matrix
\begin{eqnarray}\nonumber
M_n =\left(
  \begin{array}{cccc}
    a & b & 0 & c \\
    b & d+\frac1{n} & c & 0 \\
    0 & c & a & -b \\
    c & 0 & -b & d+\frac1{n}
  \end{array}
  \right),
\end{eqnarray}
the determinant of which is $|M_n| = a^2/n^2>0$, so $M_n$ is invertible.
For the characteristic functions of $X^{(n)}$ and $X$, denoted by $\phi_{X^{(n)}}$ and $\phi_{X}$ respectively, we have
\begin{equation}\nonumber
\begin{split}
\phi_{X^{(n)}}(\xi) & = \E( \exp(i \xi^T X^{(n)})) = \exp\left( -\frac1{2}\xi^T M_n \xi \right)
= \exp\left(-\frac1{2}\xi^T M \xi - \frac1{2n}(\xi_2^2 + \xi_4^2) \right) \\
& \longrightarrow \exp\left(-\frac1{2}\xi^T M \xi \right) = \phi_{X}(\xi), \quad n
\longrightarrow +\infty, \quad \xi \in \R^4.
\end{split}
\end{equation}
This is equivalent to the \emph{weak} convergence of probability
measures $P_{X^{(n)}} \rightarrow P_X$ as $n \rightarrow + \infty$
\cite[Theorem 7.6]{Billingsley1}, which means that $\int h(x) P_{X^{(n)}}(dx)
\longrightarrow \int h(x) P_{X}(dx)$ for all bounded and continuous
functions $h$ defined on $\R^4$.

If we define the random variable $\wt X^{(n)}$ by
\begin{equation}\nonumber
\wt X^{(n)}(\omega)  = \left\{ \begin{array}{ll}
             X^{(n)}(\omega), & \omega \in \Omega \setminus (X^{(n)})^{-1}(N),  \\
             x_0 & \omega \in (X^{(n)})^{-1}(N).
           \end{array} \right.
\end{equation}
with $N$ defined by \eqref{Ndef} and $x_0 \in \R^4 \setminus N$ fixed for all $n>0$, then Lemma \ref{nollmlem1} again gives $P_{X^{(n)}} = P_{\wt X^{(n)}}$ for all $n>0$, so we have $P_{\wt X^{(n)}} \rightarrow P_{\wt X}$ weakly.
Let us define $Z=f^{-1}(\wt X)$ and $Z^{(n)}=f^{-1}(\wt X^{(n)})$ where $Z=(Z_1,Z_2,Z_3,Z_4)$ and $Z^{(n)}=(Z_1^{(n)},Z_2^{(n)},Z_3^{(n)},Z_4^{(n)})$.
Then we have $P_Y=P_{Z_4}$ as in the first part of the proof.
Since $f^{-1}$ is continuous on
the range spaces of $\wt X$ and $\wt X^{(n)}$,
and $P_{\wt X^{(n)}} \rightarrow P_{\wt X}$ weakly, we may conclude that $P_{Z^{(n)}} \rightarrow P_{Z}$ weakly \cite[Section 1.5]{Billingsley1}. The weak convergence $P_{Z^{(n)}} \rightarrow
P_{Z}$ is equivalent to the limit of characteristic functions
\cite{Billingsley1}
\begin{equation}\label{charac1}
\phi_Z(\xi) = \lim_{n \rightarrow \infty} \phi_{Z^{(n)}} (\xi)
\quad \forall \xi \in \R^4.
\end{equation}
In the following we will compute $\phi_{Z^{(n)}} (\xi)$ using the probability density
computed in \eqref{pdf1} for invertible covariance matrix $M$. Since
$|M_n|^{1/2} = a/n$ we obtain from \eqref{pdf1} the pdf of $Z^{(n)}$ as
\begin{equation}\label{pdf2}
p_{Z^{(n)}}(z) = \frac{n}{a(2\pi)^2} z_1^2 \exp\left(- \frac{n}{2a}
\left( z_1^2(d+1/n+a z_4^2 -2 b z_4) +a z_3^2 -2c z_1 z_3 \right)
\right) \chi_{[0,+\infty)}(z_1) \chi_{(-\pi,\pi]}(z_2).
\end{equation}
It depends trivially on $z_2$, so we may concentrate on the $\R^3$-valued random variable $U^{(n)}=(Z_1^{(n)}, Z_3^{(n)}, Z_4^{(n)})$. The marginal probability density for $U^{(n)}$ is
\begin{equation}\nonumber
\begin{aligned}
p_{U^{(n)}}(z_1,z_3,z_4) & = \int_{-\pi}^\pi p_{Z^{(n)}}(z_1,z_2,z_3,z_4) d z_2 \\
& =  \frac{n}{2\pi a} z_1^2 \exp\left(- \frac{n}{2a}
\left( z_1^2(d+1/n+a z_4^2 -2 b z_4) +a z_3^2 -2c z_1 z_3 \right)
\right) \chi_{[0,+\infty)}(z_1).
\end{aligned}
\end{equation}
Thus the characteristic function of $U^{(n)}$ is
\begin{equation}\nonumber
\begin{split}
\phi_{U^{(n)}} (\xi) & = \phi_{U^{(n)}} (\xi_1,\xi_3,\xi_4) = \E( \exp(i (\xi_1,\xi_3,\xi_4)^T U^{(n)}) ) \\
& = \iiint_{\R^3} p_{U^{(n)}}(z_1,z_3,z_4) \exp( i (z_1 \xi_1 + z_3 \xi_3 + z_4 \xi_4) ) d z_1 d z_3 d z_4  \\
& = \frac{n}{2\pi a}  \int_0^{+\infty}
z_1^2 \exp \left(  - \frac{n}{2a} z_1^2 \left(d+\frac1{n}\right) + i z_1 \xi_1  \right) \\
& \quad \times \left( \int_{-\infty}^{+\infty} \int_{-\infty}^{+\infty} \exp \left( - \frac{n}{2} z_1^2 \left( z_4^2-\frac{2b}{a}z_4 \right) +i z_4 \xi_4 -\frac{n}{2} \left( z_3^2-\frac{2c}{a} z_1 z_3 \right) + i z_3 \xi_3 \right) d z_3 d z_4 \right) d z_1 \\
& = a^{-1} \exp\left( i \xi_4 \frac{b}{a} - \frac{\xi_3^2}{2n} \right)
\int_0^{+\infty} z_1 \exp \left( -\frac{z_1^2}{2a} + i z_1 \left( \xi_1+\xi_3 \frac{c}{a} \right) - \frac{\xi_4^2}{2 n z_1^2} \right) d z_1, \quad (\xi_1,\xi_3,\xi_4) \in \R^3.
\end{split}
\end{equation}
Set $U=(Z_{1}, Z_{3}, Z_{4})$. It follows from \eqref{charac1} and dominated convergence that
\begin{equation}\label{charac2}
\begin{aligned}
\phi_U (\xi_1,\xi_3,\xi_4) & = \phi_Z(\xi_1,0,\xi_3,\xi_4) = \lim_{n \longrightarrow \infty} \phi_{Z^{(n)}}
(\xi_1,0,\xi_3,\xi_4) = \lim_{n \longrightarrow \infty} \phi_{U^{(n)}} (\xi_1,\xi_3,\xi_4) \\
& = a^{-1} \exp\left( i \xi_4 \frac{b}{a} \right)
\int_0^{+\infty} z_1 \exp \left( -\frac{z_1^2}{2a} + i z_1 \left( \xi_1+\xi_3 \frac{c}{a} \right) \right) d z_1, \quad (\xi_1,\xi_3,\xi_4) \in \R^3.
\end{aligned}
\end{equation}
From \eqref{charac2} we conclude that the characteristic function
for $Z_4$ is
$$
\phi_{Z_4}(\xi_4) = \phi_U (0,0,\xi_4) = a^{-1} \exp\left( i \xi_4 \frac{b}{a} \right)
\int_0^{+\infty} z_1 \exp \left( -\frac{z_1^2}{2a} \right) d z_1 = \exp\left( i \xi_4 \frac{b}{a} \right),
$$
which implies that the probability measure for $Z_4$ is
$P_{Z_4} = \delta_{b/a} $. Finally this gives $P_Y = P_{Z_4}
= \delta_{b/a} \Rightarrow Y=b/a$ a.s.
\end{proof}

\begin{rem}\label{millerrem1}
If $z$ is WSS then $r_x(t,s)=\rho_x(t-s)$ for some even function $\rho_x$, which implies $\partial \rho_x(0)=0 \Longrightarrow c = \partial_1 r_x (t,t)=0$. In this case the matrix $M$ is block-diagonal and this special case of Proposition \ref{grundprop1} was proved by Miller \cite{Miller1}.
\end{rem}

As a corollary to Proposition \ref{grundprop1} we obtain the following result for Gaussian stochastic processes. To formulate it we need to introduce a partition of the time axis into two disjoint sets, depending on the covariance function $r_z$ for a given process $z$.  The partition is
\begin{equation}\label{Tdef}
\begin{aligned}
T & := \{ t \in \R: r_x(t,t) \partial_1 \partial_2 r_x(t,t) -
(\partial_1 r_{yx}(t,t))^2 - (\partial_1 r_{x}(t,t))^2 = 0 \}, \\
\Longrightarrow T^c=\R \setminus T & = \{ t \in \R: r_x(t,t) \partial_1 \partial_2 r_x(t,t) -
(\partial_1 r_{yx}(t,t))^2 - (\partial_1 r_{x}(t,t))^2 > 0 \},
\end{aligned}
\end{equation}
since $r_x(t,t) \partial_1 \partial_2 r_x(t,t) - (\partial_1 r_{yx}(t,t))^2 - (\partial_1 r_{x}(t,t))^2 \geq 0$ holds for all $t \in \R$ due to \eqref{parametrar1} and \eqref{parameterolikhet1}. It is clear that $T \subseteq \R$ is closed since it is the inverse image of $\{0\}$ of a continuous function. We further subdivide $T=T' \cup T''$, where $T'' \subseteq \R$ is closed, as a disjoint union of the measureable sets defined by
\begin{equation}\nonumber
\begin{aligned}
T' & := \{ t \in T:  r_x(t,t)>0 \}, \\
T'' & := \{ t \in T: r_x(t,t)=0 \}.
\end{aligned}
\end{equation}

\begin{cor}\label{gaussprop1}
Let $z(t)=x(t)+i y(t)=|z(t)| e^{i \varphi(t)}$ be a zero-mean proper Gaussian stochastic process which is differentiable according to Definition \ref{differentiable1} and let the instantaneous frequency stochastic process $\dot{\varphi}(t)$ be defined by \eqref{if1} for all $t \in \R$. Fix $t \in \R$, define $a,b,c,d$ by \eqref{parametrar1} and the matrix $M$ by \eqref{Mdef}. Then $\dot{\varphi}(t)$ has pdf
\begin{equation}\label{ifpdf1}
p_{\dot{\varphi}(t)}(y) =
\left\{
  \begin{array}{ll}
    \frac{a}{2} |M|^{1/2} \left(  (a y -b)^2 + |M|^{1/2} \right)^{-3/2} & \mbox{if } t \in \R \setminus T, \\
    \delta_{b/a}(y)  & \mbox{if } t \in T', \\
\end{array}
\right.
\end{equation}
and if $t \in T''$ then $\dot{\varphi}(t)=+\infty$ a.s. Consequently $\dot{\varphi}(t)$ has mean
\begin{equation}\nonumber
\E \dot{\varphi}(t) =
\left\{
  \begin{array}{ll}
    \frac{b}{a} & \mbox{if } t \in \R \setminus T'', \\
    +\infty & \mbox{if } t \in T'',
\end{array}
\right.
\end{equation}
and variance
\begin{equation}\nonumber
\E (\dot{\varphi}(t) - \E \dot{\varphi}(t))^2 =
\left\{
  \begin{array}{ll}
    +\infty & \mbox{if } t \in \R \setminus T, \\
    0 & \mbox{if } t \in T', \\
    \mbox{undefined} & \mbox{if } t \in T''.
\end{array}
\right.
\end{equation}
\end{cor}

Finally we restrict to harmonizable processes whose spectral measure satisfies \eqref{spekmom1}, which admits a connection to the Wigner spectrum as follows.

\begin{thm}\label{gaussprop2}
Let $z(t)=x(t)+iy(t)=|z(t)| e^{i \varphi(t)}$ be a proper Gaussian harmonizable stochastic process whose spectral measure satisfies \eqref{spekmom1} and let the instantaneous frequency stochastic process $\dot{\varphi}(t)$ be defined by \eqref{if1} for all $t \in \R$.
Fix $t \in \R$. Then $\dot{\varphi}(t)$ has pdf \eqref{ifpdf1} where $a,b,c,d$ is defined by \eqref{parametrar1} and the matrix $M$ by \eqref{Mdef}.
The process $\dot{\varphi}(t)$ has mean
\begin{equation}\nonumber
\E \dot{\varphi}(t) =
\left\{
  \begin{array}{ll}
    \frac{\int_{\xi \in \R} \xi \wt W_z(t,d\xi)}{\int_{\xi \in \R} \wt W_z(t,d\xi)} & \mbox{if } t \in \R \setminus T'', \\
    +\infty & \mbox{if } t \in T'',
\end{array}
\right.
\end{equation}
and variance
\begin{equation}\nonumber
\E (\dot{\varphi}(t) - \E \dot{\varphi}(t))^2 =
\left\{
  \begin{array}{ll}
    +\infty & \mbox{if } t \in \R \setminus T, \\
    0 & \mbox{if } t \in T', \\
    \mbox{undefined} & \mbox{if } t \in T''.
\end{array}
\right.
\end{equation}
\end{thm}

\begin{proof}
First we observe that \eqref{proper1}, \eqref{namnare2} and \eqref{parametrar1} give
\begin{equation}\nonumber
a = r_x(t,t) = \frac1{2} \ \re r_z(t,t) = \frac1{2} r_z(t,t) = \frac1{4 \pi} \int_{\xi \in \R} \wt W_z(t,d\xi).
\end{equation}
Next, from the requirement \eqref{spekmom1} it follows that \eqref{spektralintegral2} may be differentiated with respect to $\tau$ under the integral:
\begin{equation}\nonumber
\frac{\partial}{\partial \tau} \left( r_z \circ \kappa \right) (t,\tau) = \iint_{\R^2} i \xi e^{i (\tau \xi + t \eta)} m_z
\circ \kappa (d\xi,d\eta).
\end{equation}
Using \eqref{wvs2}, this gives
\begin{equation}\label{delresultat1}
\frac{\partial}{\partial \tau} \left( r_z \circ \kappa \right) (t,0) = i \iint_{\R^2} \xi e^{i t \eta} m_z
\circ \kappa (d\xi,d\eta) = \frac{i}{2 \pi} \int_{\xi \in \R} \xi \wt W_z(t,d\xi).
\end{equation}
On the other hand \eqref{proper1} implies $\partial_1 r_x(t,t)=\partial_2 r_x(t,t)$ and $\partial_2 r_{yx}(t,t)=-\partial_1 r_{yx}(t,t)$, which gives
\begin{equation}\label{delresultat2}
\begin{aligned}
\left. \frac{\partial}{\partial \tau} \left( r_z \circ \kappa \right) (t,\tau) \right|_{\tau=0} & =
\left. 2 \frac{\partial}{\partial \tau} \left( r_x \circ \kappa + i r_{yx} \circ \kappa \right) (t,\tau) \right|_{\tau=0} \\
& = \partial_1 r_x(t,t)-\partial_2 r_x(t,t) + i \left( \partial_1 r_{yx}(t,t) - \partial_2 r_{yx}(t,t) \right)  \\
& = 2 i \partial_1 r_{yx}(t,t).
\end{aligned}
\end{equation}
Combining \eqref{delresultat1}, \eqref{delresultat2} and \eqref{parametrar1} we have proved
\begin{equation}\nonumber
4 \pi b = \int_{\xi \in \R} \xi \wt W_z(t,d\xi).
\end{equation}
The result now follows from Corollary \ref{gaussprop1}, which in particular gives
\begin{equation}\label{vvinstfreq1}
\E \dot{\varphi}(t) = \frac{b}{a} = \frac{\int_{\xi \in \R} \xi
\wt W_z(t,d\xi)}{\int_{\xi \in \R} \wt W_z(t,d\xi)}
\end{equation}
provided $t \in \R \setminus T''$.
\end{proof}

\section{Cases of constantly zero or constantly infinite variance IF}\label{cases}

The process $\dot{\varphi}(t)$ exhibits completely different behavior on $t \in T'=T\setminus T''$ (where it is equal to $b/a$ with probability one), and $t \in \R \setminus T$ (where it has infinite variance with mean value $b/a$). It is therefore of interest to investigate questions like necessary or sufficient conditions for $T=\emptyset$ or $T=\R$. We restrict to harmonizable processes whose spectral measure satisfies \eqref{spekmom1}. First we look at WSS processes.

\begin{prop}\label{wssprop1}
Suppose $z$ satisfies the requirements of Theorem \ref{gaussprop2} and $z$ is WSS with covariance function $r_z(t,s)=\rho_z(t-s)$ and nonzero. Then we have $T''=\emptyset$, and either $T=\emptyset$ or $T=T'=\R$. In the latter case $\rho_z=2 \rho_x+ 2i \rho_{yx}$ where $\rho_x(t)=\alpha \cos( \beta t)$, $\alpha>0$ and $\beta \geq 0$.
\end{prop}

\begin{proof}
The assumption that $z$ is nonzero means that $\rho_x(0)>0$.
As observed in Remark \ref{millerrem1} we have $c = \partial_1 r_x (t,t)=0$ if $z$ is WSS.
Since $\partial_1 \partial_2 r_x(t,s)= -\partial^2 \rho_x(t-s)$ we have
$$
T = \{ t \in \R: r_x(t,t) \partial_1 \partial_2 r_x(t,t) =
(\partial_1 r_{yx}(t,t))^2 \} = \{ t \in \R: -\rho_x(0) \partial^2 \rho_x(0) =
(\partial \rho_{yx}(0))^2 \}.
$$
It is thus clear that either $T=\R$ or $T=\emptyset$, depending on whether $-\rho_x(0) \partial^2 \rho_x(0) =
(\partial \rho_{yx}(0))^2$ or $-\rho_x(0) \partial^2 \rho_x(0) > (\partial \rho_{yx}(0))^2$ is satisfied. If $T=\emptyset$ then of course $T''=\emptyset$.
Thus suppose that $T=\R$, i.e. $-\rho_x(0) \partial^2 \rho_x(0) = (\partial \rho_{yx}(0))^2$. This means that
$$
\E x(t) x(t) \E \dot{x}(t) \dot{x}(t) = (\E \dot{y}(t) x(t))^2 \leq \E x(t) x(t) \E \dot{y}(t) \dot{y}(t) = \E x(t) x(t) \E \dot{x}(t) \dot{x}(t), \quad t \in \R,
$$
where the inequality is the Cauchy--Schwarz inequality, and where $\E \dot{x}(t) \dot{x}(t)=\E \dot{y}(t) \dot{y}(t)=\partial_1 \partial_2 r_x(t,t)$ since $r_x=r_y$. Therefore we have equality in Cauchy--Schwarz, $\E x(t) x(t) \E \dot{y}(t) \dot{y}(t) = (\E \dot{y}(t) x(t))^2$ for all $t \in \R$. This means that $\dot{y}(t)$ equals a real multiple of $x(t)$ as a member of $L_0^2(\Omega)$, that is $\dot{y}(t) = c(t) x(t)$ where $c(t) \in \R$.
This gives $r_{\dot{y}}(t,s) = \partial_1 \partial_2 r_y (t,s)=- \partial^2 \rho_x(t-s) = c(t) c(s) \rho_x(t-s)$ $\Longrightarrow c(t) = \pm \sqrt{ - \partial^2 \rho_x(0)/\rho_x(0)}$, that is $c(t)$ is constant.
We obtain the differential equation $- \partial^2 \rho_x(t)=c^2 \rho_x(t)$ with solution $\rho_x(t)=\alpha \cos( c t)= \alpha \cos(|c| t)$ (since $\rho_x$ is even) where $\alpha=\rho_x(0)=\E x(t)^2>0$, and $\beta=|c| \geq 0$. Finally, $T=T'$, that is $T''=\emptyset$, since $\rho_x(0)>0$.
\end{proof}

\begin{example}\label{finiteexpo1}
Consider the process $z(t)=X_1 e^{i t \xi} + X_2 e^{i t \eta}$ where $\xi, \eta \in \R$, $\xi \neq \eta$ and $X_1$, $X_2$ are proper, independent zero-mean Gaussians, i.e. $\E X_1 \overline{X_2}=\E X_1 X_2=\E X_1^2 = \E X_2^2 = 0$. The process $z$ is proper, and WSS because the covariance function is $r_z(t,s)=\E|X_1|^2 e^{i \xi (t-s)} + \E|X_2|^2 e^{i \eta (t-s)}$. The real part is $\re r_z(t,s) = 2 r_x(t,s) = 2 \rho_x(t-s) = \E|X_1|^2 \cos(\xi (t-s)) + \E|X_2|^2 \cos(\eta (t-s))$. From Proposition \ref{wssprop1} we may conclude that $T=\emptyset$, that is, the IF process $\dot{\varphi}(t)$ has infinite variance for all $t \in \R$. So a linear combination of pure exponential functions, with independent proper Gaussian weights, has an infinite-variance IF process everywhere unless it consists of a single term, in which case it follows from Remark \ref{singlefunction1} that its IF process is deterministic (with variance zero for all time points).
\end{example}

Next we study the class of harmonizable processes which is larger than the class of mean-square continuous WSS processes. By the following three examples we show that $T=\R$ or $T=\emptyset$ may occur for harmonizable processes that are not WSS. We do not know whether $\emptyset \subsetneq T \subsetneq \R$ may occur.

\begin{example}
Let $x(t)$ and $y(t)$ be independent Gaussian real-valued processes with identical \emph{locally stationary} \cite{Flandrin1,Silverman1} covariance function
$$
r_x(t,s)=r_y(t,s) = \exp \left( -2 a \left( \frac{t+s}{2} \right)^2 - \frac{b}{2} (t-s)^2 \right),
$$
which is a covariance function provided $b \geq a \geq 0$ \cite{Silverman1}. We compute $\partial_1 \partial_2 r_x(t,t) = ((b-a) + 4 a^2 t^2 ) \exp(-2 a t^2)$ and $\partial_1 r_x(t,t) = - 2 a t \exp(- 2 a t^2)$ which gives
$$
r_x(t,t) \partial_1 \partial_2 r_x(t,t) - (\partial_1 r_x(t,t))^2 = ( b-a + 4 a^2 t^2 - 4 a^2 t^2 ) \exp(- 4 a t^2) = ( b-a ) \exp(- 4 a t^2) > 0, \quad t \in \R,
$$
if $b>a$. Thus by \eqref{Tdef} $T=\emptyset$ in this case.
\end{example}

\begin{example}
Let $g \in C^1(\R)$ be real-valued and the Fourier transform of a bounded measure, and let again $x(t)$ and $y(t)$ be independent Gaussian real-valued processes with equal covariance function $r_x(t,s) = r_y(t,s) = g(t) g(s)$. If $z(t)=x(t)+iy(t)$ then $r_z(t,s)=2 g(t) g(s)$. We have
$$
r_x(t,t) \partial_1 \partial_2 r_x(t,t) - (\partial_1 r_x(t,t))^2 = g(t)^2 \dot{g}(t)^2 - ( \dot{g}(t) g(t) )^2 = 0, \quad t \in \R,
$$
which means that $T=\R$ where $T$ is defined by  \eqref{Tdef}. In this example $|r_z(t,s)|^2 = r_z(t,t) r_z(s,s)$, i.e. we have equality in Cauchy--Schwarz for all $t,s \in \R$, which implies that there exists $s_0 \in \R$ such that   $z(t)=c(t)z(s_0)$ where $c(t)$ is a deterministic function. In fact $c(t)$ is a multiple of $g(t)$. Thus we are in the situation of Remark \ref{singlefunction1}.
\end{example}

\begin{example}\label{twononstat}
This example is a generalization of Example \ref{finiteexpo1}.
Let $z(t)=X_1 e^{i t \xi} + X_2 e^{i t \eta}$ where $\xi, \eta \in \R$, $\xi \neq \eta$ and $X_1$, $X_2$ are jointly proper, zero-mean, unit variance Gaussians, i.e. $\E|X_1|^2 = \E|X_2|^2 = 1$ and $\E X_1 X_2=\E X_1^2 = \E X_2^2 = 0$. In contrast to Example \ref{finiteexpo1} we allow a non-zero correlation between $X_1$ and $X_2$, $\E X_1 \overline{X_2} := c \neq 0$. We assume $|c|<1$ since otherwise we have equality $|\E X_1 \overline{X_2} |^2 = \E|X_1|^2 \E|X_2|^2$ in the Cauchy--Schwarz inequality, which implies that $X_2$ is a complex multiple of $X_1$, reducing the problem to that of Remark \ref{singlefunction1}.

The process $z$ is proper, harmonizable but not WSS because the covariance function is $r_z(t,s)=e^{i \xi (t-s)} + e^{i \eta (t-s)} + c e^{i(\xi t - \eta s)} + \overline{c} e^{i(\eta t - \xi s)}$. The terms in the decomposition $r_z=2 r_x + 2 i r_{yx}$ are
\begin{equation}\nonumber
\begin{aligned}
2 r_x(t,s) & = \cos(\xi (t-s)) + \cos(\eta (t-s)) \\
& + \re c \left( \cos( \xi t - \eta s) + \cos(\eta t-\xi s) \right)
+ \im c \left( \sin(\eta t-\xi s) - \sin(\xi t-\eta s) \right), \\
2 r_{yx}(t,s) & = \sin(\xi (t-s)) + \sin(\eta (t-s)) \\
& + \re c \left( \sin( \xi t - \eta s) + \sin(\eta t-\xi s) \right) + \im c \left( \cos(\xi t-\eta s) - \cos(\eta t-\xi s) \right).
\end{aligned}
\end{equation}
Straightforward computations give
\begin{equation}\nonumber
\begin{aligned}
2 r_x(t,t) & = 2( 1 + \re c \cos(t(\xi-\eta)) - \im c \sin(t(\xi-\eta)) ),  \\
2 \partial_1 \partial_2 r_x(t,t) & = \xi^2 + \eta^2 + 2 \xi \eta \left( \re c \cos(t(\xi-\eta)) - \im c \sin(t(\xi-\eta)) \right), \\
2 \partial_1 r_x(t,t) & = (\eta-\xi) \left( \re c \sin(t(\xi-\eta)) + \im c \cos(t(\xi-\eta)) \right), \\
2 \partial_1 r_{yx}(t,t) & = (\eta+\xi) \left( 1 + \re c \cos(t(\xi-\eta)) - \im c \sin(t(\xi-\eta)) \right),
\end{aligned}
\end{equation}
which in turn yields
\begin{equation}\nonumber
\begin{aligned}
4 \left(r_x(t,t) \partial_1 \partial_2 r_x(t,t) -
(\partial_1 r_{yx}(t,t))^2 - (\partial_1 r_{x}(t,t))^2 \right) = (\xi-\eta)^2 (1-|c|^2) > 0
\end{aligned}
\end{equation}
for all $t \in \R$. This means that $T=\emptyset$. We conclude that a linear combination of two pure exponential functions, with correlated proper Gaussian weights of equal power, is a nonstationary process whose IF process $\dot{\varphi}(t)$ has infinite variance everywhere.
\end{example}

\section{Conclusion}

We have derived the pdf of the IF of a mean-square differentiable proper Gaussian stochastic process at a fixed and arbitrary time point. The proof does not involve the frequency domain. It shows that the IF has either infinite or zero variance. Moreover, if we assume that the process is harmonizable then we obtain as a consequence of our main result the formula for the mean of the IF as a normalized first-order frequency moment of the Wigner spectrum, again for fixed time. This is a generalization of the corresponding well-known formula for deterministic functions.

The question of the dichotomy zero/infinite variance of the IF process is left somewhat open in this paper. We do not know if one process may exhibit nonempty time sets of both kinds. We have shown by examples that IF processes with constant zero variance exist, and IF processes with constant infinite variance exist, and each of these cases can occur both inside and outside the class of WSS processes.

\section*{Acknowledgement}

We express our gratitude to Bj\"orn R\"uffer for his helpful comments and remarks.

\end{document}